\def\papertitle{An Aliasing-Free Hybrid Digital-Analog Polyphonic Synthesizer}
\def\paperauthorA{Jonas Roth}
\def\paperauthorB{Domenic Keller}
\def\paperauthorC{Oscar Casta\~neda}
\def\paperauthorD{Christoph Studer}

\documentclass[twoside,a4paper]{article}
\usepackage{etoolbox}
\usepackage{macros/dafx_23}
\usepackage{amsmath,amssymb,amsfonts,amsthm}
\usepackage{euscript}
\usepackage[T1]{fontenc}
\usepackage[utf8]{inputenc}
\usepackage{ifpdf}
\usepackage[english]{babel}
\usepackage{caption}
\usepackage{balance}
\usepackage{color}
\usepackage{booktabs}
\usepackage{subfigure}
\usepackage{microtype}

\input glyphtounicode
\ninept

\usepackage{amssymb}
\usepackage{amsfonts}
\usepackage{mathrsfs}
\usepackage{xspace}
\usepackage{bm}
\usepackage{upgreek}

\newcommand{\safemath}[2]{\newcommand{#1}{\ensuremath{#2}\xspace}}

\safemath{\bma}{\mathbf{a}}
\safemath{\bmb}{\mathbf{b}}
\safemath{\bmc}{\mathbf{c}}
\safemath{\bmd}{\mathbf{d}}
\safemath{\bme}{\mathbf{e}}
\safemath{\bmf}{\mathbf{f}}
\safemath{\bmg}{\mathbf{g}}
\safemath{\bmh}{\mathbf{h}}
\safemath{\bmi}{\mathbf{i}}
\safemath{\bmj}{\mathbf{j}}
\safemath{\bmk}{\mathbf{k}}
\safemath{\bml}{\mathbf{l}}
\safemath{\bmm}{\mathbf{m}}
\safemath{\bmn}{\mathbf{n}}
\safemath{\bmo}{\mathbf{o}}
\safemath{\bmp}{\mathbf{p}}
\safemath{\bmq}{\mathbf{q}}
\safemath{\bmr}{\mathbf{r}}
\safemath{\bms}{\mathbf{s}}
\safemath{\bmt}{\mathbf{t}}
\safemath{\bmu}{\mathbf{u}}
\safemath{\bmv}{\mathbf{v}}
\safemath{\bmw}{\mathbf{w}}
\safemath{\bmx}{\mathbf{x}}
\safemath{\bmy}{\mathbf{y}}
\safemath{\bmz}{\mathbf{z}}
\safemath{\bmzero}{\mathbf{0}}
\safemath{\bmone}{\mathbf{1}}

\bmdefine{\biad}{a}
\bmdefine{\bibd}{b}
\bmdefine{\bicd}{c}
\bmdefine{\bidd}{d}
\bmdefine{\bied}{e}
\bmdefine{\bifd}{f}
\bmdefine{\bigd}{g}
\bmdefine{\bihd}{h}
\bmdefine{\biid}{i}
\bmdefine{\bijd}{j}
\bmdefine{\bikd}{k}
\bmdefine{\bild}{l}
\bmdefine{\bimd}{m}
\bmdefine{\bind}{n}
\bmdefine{\biod}{o}
\bmdefine{\bipd}{p}
\bmdefine{\biqd}{q}
\bmdefine{\bird}{r}
\bmdefine{\bisd}{s}
\bmdefine{\bitd}{t}
\bmdefine{\biud}{u}
\bmdefine{\bivd}{v}
\bmdefine{\biwd}{w}
\bmdefine{\bixd}{x}
\bmdefine{\biyd}{y}
\bmdefine{\bizd}{z}

\bmdefine{\bixid}{\xi}
\bmdefine{\bilambdad}{\lambda}
\bmdefine{\bimud}{\mu}
\bmdefine{\bithetad}{\theta}
\bmdefine{\biphid}{\phi}
\bmdefine{\bideltad}{\delta}

\safemath{\bmia}{\biad}
\safemath{\bmib}{\bibd}
\safemath{\bmic}{\bicd}
\safemath{\bmid}{\bidd}
\safemath{\bmie}{\bied}
\safemath{\bmif}{\bifd}
\safemath{\bmig}{\bigd}
\safemath{\bmih}{\bihd}
\safemath{\bmii}{\biid}
\safemath{\bmij}{\bijd}
\safemath{\bmik}{\bikd}
\safemath{\bmil}{\bild}
\safemath{\bmim}{\bimd}
\safemath{\bmin}{\bind}
\safemath{\bmio}{\biod}
\safemath{\bmip}{\bipd}
\safemath{\bmiq}{\biqd}
\safemath{\bmir}{\bird}
\safemath{\bmis}{\bisd}
\safemath{\bmit}{\bitd}
\safemath{\bmiu}{\biud}
\safemath{\bmiv}{\bivd}
\safemath{\bmiw}{\biwd}
\safemath{\bmix}{\bixd}
\safemath{\bmiy}{\biyd}
\safemath{\bmiz}{\bizd}

\safemath{\bmxi}{\bixid}
\safemath{\bmlambda}{\bilambdad}
\safemath{\bmmu}{\bimud}
\safemath{\bmtheta}{\bithetad}
\safemath{\bmphi}{\biphid}
\safemath{\bmdelta}{\bideltad}

\safemath{\bA}{\mathbf{A}}
\safemath{\bB}{\mathbf{B}}
\safemath{\bC}{\mathbf{C}}
\safemath{\bD}{\mathbf{D}}
\safemath{\bE}{\mathbf{E}}
\safemath{\bF}{\mathbf{F}}
\safemath{\bG}{\mathbf{G}}
\safemath{\bH}{\mathbf{H}}
\safemath{\bI}{\mathbf{I}}
\safemath{\bJ}{\mathbf{J}}
\safemath{\bK}{\mathbf{K}}
\safemath{\bL}{\mathbf{L}}
\safemath{\bM}{\mathbf{M}}
\safemath{\bN}{\mathbf{N}}
\safemath{\bO}{\mathbf{O}}
\safemath{\bP}{\mathbf{P}}
\safemath{\bQ}{\mathbf{Q}}
\safemath{\bR}{\mathbf{R}}
\safemath{\bS}{\mathbf{S}}
\safemath{\bT}{\mathbf{T}}
\safemath{\bU}{\mathbf{U}}
\safemath{\bV}{\mathbf{V}}
\safemath{\bW}{\mathbf{W}}
\safemath{\bX}{\mathbf{X}}
\safemath{\bY}{\mathbf{Y}}
\safemath{\bZ}{\mathbf{Z}}

\safemath{\bZero}{\mathbf{0}}
\safemath{\bOne}{\mathbf{1}}
\safemath{\bDelta}{\mathbf{\Delta}}
\safemath{\bLambda}{\mathbf{\UpLambda}}
\safemath{\bPhi}{\mathbf{\Upphi}}
\safemath{\bSigma}{\mathbf{\Upsigma}}
\safemath{\bOmega}{\mathbf{\Upomega}}
\safemath{\bTheta}{\mathbf{\Uptheta}}

\bmdefine{\biAd}{A}
\bmdefine{\biBd}{B}
\bmdefine{\biCd}{C}
\bmdefine{\biDd}{D}
\bmdefine{\biEd}{E}
\bmdefine{\biFd}{F}
\bmdefine{\biGd}{G}
\bmdefine{\biHd}{H}
\bmdefine{\biId}{I}
\bmdefine{\biJd}{J}
\bmdefine{\biKd}{K}
\bmdefine{\biLd}{L}
\bmdefine{\biMd}{M}
\bmdefine{\biOd}{N}
\bmdefine{\biPd}{O}
\bmdefine{\biQd}{P}
\bmdefine{\biRd}{R}
\bmdefine{\biSd}{S}
\bmdefine{\biTd}{T}
\bmdefine{\biUd}{U}
\bmdefine{\biVd}{V}
\bmdefine{\biWd}{W}
\bmdefine{\biXd}{X}
\bmdefine{\biYd}{Y}
\bmdefine{\biZd}{Z}

\bmdefine{\biDelta}{\Delta}
\bmdefine{\biLambda}{\Lambda}
\bmdefine{\biPhi}{\Phi}
\bmdefine{\biSigma}{\Sigma}
\bmdefine{\biOmega}{\Omega}
\bmdefine{\biTheta}{\Theta}

\safemath{\bimA}{\biAd}
\safemath{\bimB}{\biBd}
\safemath{\bimC}{\biCd}
\safemath{\bimD}{\biDd}
\safemath{\bimE}{\biEd}
\safemath{\bimF}{\biFd}
\safemath{\bimG}{\biGd}
\safemath{\bimH}{\biHd}
\safemath{\bimI}{\biId}
\safemath{\bimJ}{\biJd}
\safemath{\bimK}{\biKd}
\safemath{\bimL}{\biLd}
\safemath{\bimM}{\biMd}
\safemath{\bimN}{\biNd}
\safemath{\bimO}{\biOd}
\safemath{\bimP}{\biPd}
\safemath{\bimQ}{\biQd}
\safemath{\bimR}{\biRd}
\safemath{\bimS}{\biSd}
\safemath{\bimT}{\biTd}
\safemath{\bimU}{\biUd}
\safemath{\bimV}{\biVd}
\safemath{\bimW}{\biWd}
\safemath{\bimX}{\biXd}
\safemath{\bimY}{\biYd}
\safemath{\bimZ}{\biZd}

\safemath{\bimDelta}{\biDelta}
\safemath{\bimLambda}{\biLambda}
\safemath{\bimPhi}{\biPhi}
\safemath{\bimSigma}{\biSigma}
\safemath{\bimOmega}{\biOmega}
\safemath{\bimTheta}{\biTheta}

\safemath{\setA}{\mathcal{A}}
\safemath{\setB}{\mathcal{B}}
\safemath{\setC}{\mathcal{C}}
\safemath{\setD}{\mathcal{D}}
\safemath{\setE}{\mathcal{E}}
\safemath{\setF}{\mathcal{F}}
\safemath{\setG}{\mathcal{G}}
\safemath{\setH}{\mathcal{H}}
\safemath{\setI}{\mathcal{I}}
\safemath{\setJ}{\mathcal{J}}
\safemath{\setK}{\mathcal{K}}
\safemath{\setL}{\mathcal{L}}
\safemath{\setM}{\mathcal{M}}
\safemath{\setN}{\mathcal{N}}
\safemath{\setO}{\mathcal{O}}
\safemath{\setP}{\mathcal{P}}
\safemath{\setQ}{\mathcal{Q}}
\safemath{\setR}{\mathcal{R}}
\safemath{\setS}{\mathcal{S}}
\safemath{\setT}{\mathcal{T}}
\safemath{\setU}{\mathcal{U}}
\safemath{\setV}{\mathcal{V}}
\safemath{\setW}{\mathcal{W}}
\safemath{\setX}{\mathcal{X}}
\safemath{\setY}{\mathcal{Y}}
\safemath{\setZ}{\mathcal{Z}}
\safemath{\emptySet}{\varnothing}

\safemath{\colA}{\mathscr{A}}
\safemath{\colB}{\mathscr{B}}
\safemath{\colC}{\mathscr{C}}
\safemath{\colD}{\mathscr{D}}
\safemath{\colE}{\mathscr{E}}
\safemath{\colF}{\mathscr{F}}
\safemath{\colG}{\mathscr{G}}
\safemath{\colH}{\mathscr{H}}
\safemath{\colI}{\mathscr{I}}
\safemath{\colJ}{\mathscr{J}}
\safemath{\colK}{\mathscr{K}}
\safemath{\colL}{\mathscr{L}}
\safemath{\colM}{\mathscr{M}}
\safemath{\colN}{\mathscr{N}}
\safemath{\colO}{\mathscr{O}}
\safemath{\colP}{\mathscr{P}}
\safemath{\colQ}{\mathscr{Q}}
\safemath{\colR}{\mathscr{R}}
\safemath{\colS}{\mathscr{S}}
\safemath{\colT}{\mathscr{T}}
\safemath{\colU}{\mathscr{U}}
\safemath{\colV}{\mathscr{V}}
\safemath{\colW}{\mathscr{W}}
\safemath{\colX}{\mathscr{X}}
\safemath{\colY}{\mathscr{Y}}
\safemath{\colZ}{\mathscr{Z}}

\safemath{\opA}{\mathbb{A}}
\safemath{\opB}{\mathbb{B}}
\safemath{\opC}{\mathbb{C}}
\safemath{\opD}{\mathbb{D}}
\safemath{\opE}{\mathbb{E}}
\safemath{\opF}{\mathbb{F}}
\safemath{\opG}{\mathbb{G}}
\safemath{\opH}{\mathbb{H}}
\safemath{\opI}{\mathbb{I}}
\safemath{\opJ}{\mathbb{J}}
\safemath{\opK}{\mathbb{K}}
\safemath{\opL}{\mathbb{L}}
\safemath{\opM}{\mathbb{M}}
\safemath{\opN}{\mathbb{N}}
\safemath{\opO}{\mathbb{O}}
\safemath{\opP}{\mathbb{P}}
\safemath{\opQ}{\mathbb{Q}}
\safemath{\opR}{\mathbb{R}}
\safemath{\opS}{\mathbb{S}}
\safemath{\opT}{\mathbb{T}}
\safemath{\opU}{\mathbb{U}}
\safemath{\opV}{\mathbb{V}}
\safemath{\opW}{\mathbb{W}}
\safemath{\opX}{\mathbb{X}}
\safemath{\opY}{\mathbb{Y}}
\safemath{\opZ}{\mathbb{Z}}
\safemath{\opZero}{\mathbb{O}}
\safemath{\identityop}{\opI}

\safemath{\veca}{\bma}
\safemath{\vecb}{\bmb}
\safemath{\vecc}{\bmc}
\safemath{\vecd}{\bmd}
\safemath{\vece}{\bme}
\safemath{\vecf}{\bmf}
\safemath{\vecg}{\bmg}
\safemath{\vech}{\bmh}
\safemath{\veci}{\bmi}
\safemath{\vecj}{\bmj}
\safemath{\veck}{\bmk}
\safemath{\vecl}{\bml}
\safemath{\vecm}{\bmm}
\safemath{\vecn}{\bmn}
\safemath{\veco}{\bmo}
\safemath{\vecp}{\bmp}
\safemath{\vecq}{\bmq}
\safemath{\vecr}{\bmr}
\safemath{\vecs}{\bms}
\safemath{\vect}{\bmt}
\safemath{\vecu}{\bmu}
\safemath{\vecv}{\bmv}
\safemath{\vecw}{\bmw}
\safemath{\vecx}{\bmx}
\safemath{\vecy}{\bmy}
\safemath{\vecz}{\bmz}

\safemath{\veczero}{\bmzero}
\safemath{\vecone}{\bmone}
\safemath{\vecxi}{\bmxi}
\safemath{\veclambda}{\bmlambda}
\safemath{\vecmu}{\bmmu}
\safemath{\vectheta}{\bmtheta}
\safemath{\vecphi}{\bmphi}
\safemath{\vecdelta}{\bmdelta}

\safemath{\matA}{\bA}
\safemath{\matB}{\bB}
\safemath{\matC}{\bC}
\safemath{\matD}{\bD}
\safemath{\matE}{\bE}
\safemath{\matF}{\bF}
\safemath{\matG}{\bG}
\safemath{\matH}{\bH}
\safemath{\matI}{\bI}
\safemath{\matJ}{\bJ}
\safemath{\matK}{\bK}
\safemath{\matL}{\bL}
\safemath{\matM}{\bM}
\safemath{\matN}{\bN}
\safemath{\matO}{\bO}
\safemath{\matP}{\bP}
\safemath{\matQ}{\bQ}
\safemath{\matR}{\bR}
\safemath{\matS}{\bS}
\safemath{\matT}{\bT}
\safemath{\matU}{\bU}
\safemath{\matV}{\bV}
\safemath{\matW}{\bW}
\safemath{\matX}{\bX}
\safemath{\matY}{\bY}
\safemath{\matZ}{\bZ}
\safemath{\matzero}{\bmzero}

\safemath{\matDelta}{\bDelta}
\safemath{\matLambda}{\bLambda}
\safemath{\matPhi}{\bPhi}
\safemath{\matSigma}{\bSigma}
\safemath{\matOmega}{\bOmega}
\safemath{\matTheta}{\bTheta}

\safemath{\matidentity}{\matI}
\safemath{\matone}{\matO}

\safemath{\rnda}{A}
\safemath{\rndb}{B}
\safemath{\rndc}{C}
\safemath{\rndd}{D}
\safemath{\rnde}{E}
\safemath{\rndf}{F}
\safemath{\rndg}{G}
\safemath{\rndh}{H}
\safemath{\rndi}{I}
\safemath{\rndj}{J}
\safemath{\rndk}{K}
\safemath{\rndl}{L}
\safemath{\rndm}{M}
\safemath{\rndn}{N}
\safemath{\rndo}{O}
\safemath{\rndp}{P}
\safemath{\rndq}{Q}
\safemath{\rndr}{R}
\safemath{\rnds}{S}
\safemath{\rndt}{T}
\safemath{\rndu}{U}
\safemath{\rndv}{V}
\safemath{\rndw}{W}
\safemath{\rndx}{X}
\safemath{\rndy}{Y}
\safemath{\rndz}{Z}

\safemath{\rveca}{\bimA}
\safemath{\rvecb}{\bimB}
\safemath{\rvecc}{\bimC}
\safemath{\rvecd}{\bimD}
\safemath{\rvece}{\bimE}
\safemath{\rvecf}{\bimF}
\safemath{\rvecg}{\bimG}
\safemath{\rvech}{\bimH}
\safemath{\rveci}{\bimI}
\safemath{\rvecj}{\bimJ}
\safemath{\rveck}{\bimK}
\safemath{\rvecl}{\bimL}
\safemath{\rvecm}{\bimM}
\safemath{\rvecn}{\bimN}
\safemath{\rveco}{\bomO}
\safemath{\rvecp}{\bimP}
\safemath{\rvecq}{\bimQ}
\safemath{\rvecr}{\bimR}
\safemath{\rvecs}{\bimS}
\safemath{\rvect}{\bimT}
\safemath{\rvecu}{\bimU}
\safemath{\rvecv}{\bimV}
\safemath{\rvecw}{\bimW}
\safemath{\rvecx}{\bimX}
\safemath{\rvecy}{\bimY}
\safemath{\rvecz}{\bimZ}

\safemath{\rvecxi}{\bmxi}
\safemath{\rveclambda}{\bmlambda}
\safemath{\rvecmu}{\bmmu}
\safemath{\rvectheta}{\bmtheta}
\safemath{\rvecphi}{\bmphi}

\safemath{\rmatA}{\bimA}
\safemath{\rmatB}{\bimB}
\safemath{\rmatC}{\bimC}
\safemath{\rmatD}{\bimD}
\safemath{\rmatE}{\bimE}
\safemath{\rmatF}{\bimF}
\safemath{\rmatG}{\bimG}
\safemath{\rmatH}{\bimH}
\safemath{\rmatI}{\bimI}
\safemath{\rmatJ}{\bimJ}
\safemath{\rmatK}{\bimK}
\safemath{\rmatL}{\bimL}
\safemath{\rmatM}{\bimM}
\safemath{\rmatN}{\bimN}
\safemath{\rmatO}{\bimO}
\safemath{\rmatP}{\bimP}
\safemath{\rmatQ}{\bimQ}
\safemath{\rmatR}{\bimR}
\safemath{\rmatS}{\bimS}
\safemath{\rmatT}{\bimT}
\safemath{\rmatU}{\bimU}
\safemath{\rmatV}{\bimV}
\safemath{\rmatW}{\bimW}
\safemath{\rmatX}{\bimX}
\safemath{\rmatY}{\bimY}
\safemath{\rmatZ}{\bimZ}

\safemath{\rmatDelta}{\bimDelta}
\safemath{\rmatLambda}{\bimLambda}
\safemath{\rmatPhi}{\bimPhi}
\safemath{\rmatSigma}{\bimSigma}
\safemath{\rmatOmega}{\bimOmega}
\safemath{\rmatTheta}{\bimTheta}

\usepackage{amssymb}
\usepackage{amsfonts}
\usepackage{mathrsfs}
\usepackage{xspace}
\usepackage{bm}
\usepackage{fancyref}
\usepackage{textcomp}

\usepackage{multirow}
\usepackage{stmaryrd}

\newenvironment{textbmatrix}{	\setlength{\arraycolsep}{2.5pt}%
								\big[\begin{matrix}}{\end{matrix}\big]%
								\raisebox{0.08ex}{\vphantom{M}}}

\def\be{\begin{equation}}
\def\ee{\end{equation}}
\def\een{\nonumber \end{equation}}
\def\mat{\begin{bmatrix}}
\def\emat{\end{bmatrix}}
\def\btm{\begin{textbmatrix}}
\def\etm{\end{textbmatrix}}

\def\ba#1\ea{\begin{align}#1\end{align}}
\def\bas#1\eas{\begin{align*}#1\end{align*}}
\def\bs#1\es{\begin{split}#1\end{split}}
\def\bg#1\eg{\begin{gather}#1\end{gather}}
\def\bml#1\eml{\begin{multline}#1\end{multline}}
\def\bi#1\ei{\begin{itemize}#1\end{itemize}}

\safemath{\dirac}{\delta}					
\safemath{\krond}{\dirac}					

\safemath{\upto}{\uparrow}
\safemath{\downto}{\downarrow}
\safemath{\iu}{j}							
\safemath{\ev}{\lambda}						
\safemath{\hilseqspace}{l^{2}}				
\newcommand{\banachfunspace}[1]{\setL^{#1}}	
\safemath{\hilfunspace}{\banachfunspace{2}}	

\safemath{\SNR}{\textit{SNR}} 				
\safemath{\PAR}{\textit{PAR}} 				
\safemath{\No}{N_0}							
\safemath{\Es}{E_s}							
\safemath{\Eb}{E_b}							
\safemath{\EbNo}{\frac{\Eb}{\No}}
\safemath{\EsNo}{\frac{\Es}{\No}}

\DeclareMathOperator{\CHop}{\ensuremath{\opH}} 
\safemath{\tvir}{\rndh_{\CHop}}				
\safemath{\tvtf}{\rndl_{\CHop}}				
\safemath{\spf}{\rnds_{\CHop}}				
\safemath{\bff}{H_{\CHop}}					

\safemath{\ircf}{r_{h}}						
\safemath{\tftvcf}{r_{s}}					
\safemath{\tfcf}{r_{l}}						
\safemath{\bfcf}{r_{H}}						

\safemath{\tcorr}{c_h}						
\safemath{\scf}{c_{s}}						
\safemath{\tfcorr}{c_{l}}					
\safemath{\fcorr}{c_{H}}						

\safemath{\mi}{I}							
\safemath{\capacity}{C}						

\safemath{\normal}{\mathcal{N}}			
\safemath{\jpg}{\mathcal{CN}}			
\safemath{\mchain}{\leftrightarrow}		

\safemath{\dB}{\,\mathrm{dB}}
\safemath{\dBm}{\,\mathrm{dBm}}
\safemath{\Hz}{\,\mathrm{Hz}}
\safemath{\kHz}{\,\mathrm{kHz}}
\safemath{\MHz}{\,\mathrm{MHz}}
\safemath{\GHz}{\,\mathrm{GHz}}
\safemath{\s}{\,\mathrm{s}}
\safemath{\ms}{\,\mathrm{ms}}
\safemath{\mus}{\,\mathrm{\text{\textmu}s}}
\safemath{\ns}{\,\mathrm{ns}}
\safemath{\ps}{\,\mathrm{ps}}
\safemath{\meter}{\,\mathrm{m}}
\safemath{\mm}{\,\mathrm{mm}}
\safemath{\cm}{\,\mathrm{cm}}
\safemath{\m}{\,\mathrm{m}}
\safemath{\W}{\,\mathrm{W}}
\safemath{\mW}{\, \mathrm{mW}}
\safemath{\J}{\,\mathrm{J}}
\safemath{\K}{\,\mathrm{K}}
\safemath{\bit}{\,\mathrm{bit}}
\safemath{\nat}{\,\mathrm{nat}}

\safemath{\define}{\triangleq}			

\safemath{\equivalent}{\sim}
\safemath{\distas}{\sim}					
\safemath{\sdiff}{\Delta}				

\safemath{\reals}{\mathbb{R}}
\safemath{\positivereals}{\reals_{+}}
\safemath{\integers}{\mathbb{Z}}
\safemath{\posint}{\integers_{+}}
\safemath{\naturals}{\mathbb{N}}
\safemath{\posnaturals}{\naturals_{+}}
\safemath{\complexset}{\mathbb{C}}
\safemath{\rationals}{\mathbb{Q}}

\newcommand*{\fancyrefapplabelprefix}{app}		
\newcommand*{\fancyrefthmlabelprefix}{thm}		
\newcommand*{\fancyreflemlabelprefix}{lem}		
\newcommand*{\fancyrefcorlabelprefix}{cor}		
\newcommand*{\fancyrefdeflabelprefix}{def}		
\newcommand*{\fancyrefproplabelprefix}{prop}		
\newcommand*{\fancyrefexmpllabelprefix}{exmpl}
\newcommand*{\fancyrefalglabelprefix}{alg}		
\newcommand*{\fancyreftbllabelprefix}{tbl}		

\frefformat{vario}{\fancyrefseclabelprefix}{Section~#1}
\frefformat{vario}{\fancyrefthmlabelprefix}{Theorem.~#1}
\frefformat{vario}{\fancyreftbllabelprefix}{Table~#1}
\frefformat{vario}{\fancyreflemlabelprefix}{Lemma~#1}
\frefformat{vario}{\fancyrefcorlabelprefix}{Corollary~#1}
\frefformat{vario}{\fancyrefdeflabelprefix}{Definition~#1}
\frefformat{vario}{\fancyreffiglabelprefix}{Figure~#1}
\frefformat{vario}{\fancyrefapplabelprefix}{Appendix~#1}
\frefformat{vario}{\fancyrefeqlabelprefix}{(#1)}
\frefformat{vario}{\fancyrefproplabelprefix}{Proposition~#1}
\frefformat{vario}{\fancyrefexmpllabelprefix}{Example~#1}
\frefformat{vario}{\fancyrefalglabelprefix}{Algorithm~#1}

 \newtheorem{thm}{Theorem}

 \newtheorem{lem}[thm]{Lemma}
 
 \newtheorem{remark}{Remark}
 \newtheorem*{remark*}{Remark}

\safemath{\dictab}{[\,\dicta\,\,\dictb\,]}

\safemath{\ysig}{\bmy}
\safemath{\ysighat}{\hat{\ysig}}
\safemath{\ysigdim}{M}
\safemath{\xsig}{\bmx}
\safemath{\xsigdim}{N}
\safemath{\nx}{n_x}
\safemath{\zsig}{\bmz}
\safemath{\zsigdim}{\ysigdim}
\safemath{\rsig}{\bmr}
\safemath{\Adict}{\bA}
\safemath{\Adicttilde}{\widetilde{\Adict}}
\safemath{\Adictdim}{\outputdim\times\xsigdim}
\safemath{\avec}{\bma}
\safemath{\avectilde}{\tilde{\avec}}
\safemath{\Bdict}{\bB}
\safemath{\Bdicttilde}{\widetilde{\Bdict}}
\safemath{\Cdict}{\bC}
\safemath{\cvec}{\bmc}
\safemath{\Ddict}{\bD}
\safemath{\Ddictdim}{\ysigdim\times\xsigdim}
\safemath{\dvec}{\bmd}
\safemath{\Ddicttilde}{\widetilde{\bD}}
\safemath{\Bonb}{\bB}
\safemath{\bvec}{\bmb}
\safemath{\Bonbdim}{\ysigdim\times\ysigdim}
\safemath{\noise}{\bmn}
\safemath{\noisedim}{\ysigim}
\safemath{\err}{\bme}
\safemath{\errdim}{\ysigdim}
\safemath{\errset}{\setE}
\safemath{\nerr}{n_e}
\safemath{\delop}{\bP_\errset}
\safemath{\delopc}{\bP_{{\errset}^c}}

\safemath{\cplxi}{\imath}
\safemath{\cplxj}{\jmath}

\safemath{\dict}{\matD}
\safemath{\inputdim}{N}		
\safemath{\outputdim}{M}		
\safemath{\sparsity}{S}	
\safemath{\inputdimA}{{N_a}}	
\safemath{\inputdimB}{{N_b}}	
\safemath{\elemA}{{n_a}}	
\safemath{\elemB}{{n_b}}	
\safemath{\resA}{\matR_a}	
\safemath{\resB}{\matR_b}	
\safemath{\subD}{\matS} 
\safemath{\subA}{\matS_a} 
\safemath{\subB}{\matS_b} 
\safemath{\dicta}{\matA} 	
\safemath{\dictb}{\matB} 	
\safemath{\hollowS}{H}
\safemath{\hollowA}{H_a}
\safemath{\hollowB}{H_b}
\safemath{\cross}{Z}
\safemath{\coh}{\mu_d}			
\safemath{\coha}{\mu_a}			
\safemath{\cohb}{\mu_b}			
\safemath{\mubs}{\nu}	
\safemath{\cohm}{\mu_m} 
\safemath{\dictset}{\setD}	
\safemath{\dictsetp}{\dictset(\coh,\coha,\cohb)}	
\safemath{\dictsetgen}{\dictset_\text{gen}}
\safemath{\dictsetgenp}{\dictsetgen(\coh)}
\safemath{\dictsetonb}{\dictset_\text{onb}}
\safemath{\dictsetonbp}{\dictsetonb(\coh)}

\safemath{\leftside}{U}
\safemath{\rightsideA}{R_a}
\safemath{\rightsideB}{R_b}

\safemath{\indexS}{\setI_S} 

\safemath{\na}{n_a}			
\safemath{\nb}{n_b}			
\safemath{\coeffa}{p_i}	
\safemath{\coeffb}{q_j}	
\safemath{\seta}{\setP}		
\safemath{\setb}{\setQ}     
\safemath{\setw}{\setW}	
\safemath{\setz}{\setZ}	
\safemath{\cola}{\veca}		
\safemath{\colb}{\vecb}		
\safemath{\cold}{\vecd}		
\safemath{\inputvec}{\vecx} 	
\safemath{\error}{\vece}	
\safemath{\noiseout}{\vecz} 	
\safemath{\inputvecel}{x}
\safemath{\inputveca}{\vecx_a}
\safemath{\inputvecb}{\vecx_b}
\safemath{\outputvec}{\vecy}	
\safemath{\lambdamin}{\lambda_{\mathrm{min}}}

\safemath{\elltwo}{\ell_2}
\safemath{\ellone}{\ell_1}
\safemath{\ellzero}{\ell_0}
\safemath{\ellinf}{\ell_\infty}
\safemath{\ellinftilde}{\ell_{\widetilde\infty}}
\safemath{\licard}{Z(\coh,\coha,\cohb)}
\safemath{\xsol}{\hat{x}}
\safemath{\xbord}{x_b}		
\safemath{\xstat}{x_s}		
\safemath{\xstatLone}{\tilde{x}_s}
\safemath{\order}{\mathcal{O}} 
\safemath{\scales}{\Theta} 
\safemath{\ones}{\mathbf{1}} 
\safemath{\zeroes}{\mathbf{0}} 
\safemath{\thlone}{\kappa(\coh,\cohb)} 
\safemath{\constoneA}{\delta} 
\safemath{\constoneB}{\epsilon} 
\safemath{\nlarge}{L}				   
\safemath{\sumlarge}{S_\nlarge}
\safemath{\maxlarger}{P_\nlarge}	   
\safemath{\Pzero}{\textrm{P0}}	
\safemath{\Pone}{\textrm{P1}}
\safemath{\vecfir}{\vecw}			 
\safemath{\vecsec}{\vecz}
\safemath{\elvecfir}{w}              
\safemath{\elvecsec}{z}				 
\safemath{\nlargefir}{n}
\safemath{\normout}{\gamma}
\safemath{\auxfun}{h}
\safemath{\supp}{\textrm{supp}}

\safemath{\indexa}{\ell}
\safemath{\indexb}{r}
\safemath{\indexc}{i}
\safemath{\indexd}{j}

\safemath{\project}{P}

\newcommand{\mypm}{\mathbin{\smash{%
\raisebox{0.35ex}{%
            $\underset{\raisebox{0.3ex}{$\smash -$}}{\smash+}$%
            }%
        }%
    }%
}

\newcommand{\pmsynth}{$\boldsymbol\mypm$synth\xspace} 
\newcommand{\mcu}{MCU\xspace}
\newcommand{\fs}{\ensuremath{f_\text{s}}}
\newcommand{\qf}{\ensuremath{q_\text{f}}}
\newcommand{\qi}{\ensuremath{q_\text{i}}}

\newcounter{numauth}
\newcounter{listcnt}
\newcommand\authcnt[1]{\ifdefined#1 \stepcounter{numauth} \fi}

\newcommand\addauth[1]{
\ifdefined#1 
\stepcounter{listcnt}
\ifnum \value{listcnt}<\value{numauth}
\appto\authorslist{, #1}
\else
\appto\authorslist{~and~#1}
\fi
\fi}
\authcnt{\paperauthorB}
\authcnt{\paperauthorC}
\authcnt{\paperauthorD}
\authcnt{\paperauthorE}
\authcnt{\paperauthorF}
\authcnt{\paperauthorG}
\authcnt{\paperauthorH}
\authcnt{\paperauthorI}
\authcnt{\paperauthorJ}
\def\authorslist{\paperauthorA}
\addauth{\paperauthorB}
\addauth{\paperauthorC}
\addauth{\paperauthorD}
\addauth{\paperauthorE}
\addauth{\paperauthorF}
\addauth{\paperauthorG}
\addauth{\paperauthorH}
\addauth{\paperauthorI}
\addauth{\paperauthorJ}

\usepackage{times}
\usepackage{upgreek}

\renewcommand{\paragraph}[1]{{\vspace{\baselineskip}\noindent\normalfont\bfseries#1\quad}}

\newif\ifpdf
\ifx\pdfoutput\relax
\else
   \ifcase\pdfoutput
      \pdffalse
   \else
      \pdftrue
\fi

\ifpdf 
  \usepackage[pdftex,
    pdftitle={\papertitle},
    pdfauthor={\authorslist},
    pdfsubject={Proceedings of the 26th International Conference on Digital Audio Effects (DAFx23)},
    colorlinks=false, 
    bookmarksnumbered, 
    pdfstartview=XYZ 
  ]{hyperref}
  \usepackage[pdftex]{graphicx}
\else 
  \usepackage[dvips]{epsfig,graphicx}
  \usepackage[dvips,
    pdftitle={\papertitle},
    pdfauthor={\authorslist},
    pdfsubject={Proceedings of the 26th International Conference on Digital Audio Effects (DAFx23)},
    colorlinks=false, 
    bookmarksnumbered, 
    pdfstartview=XYZ 
  ]{hyperref}
\fi
\usepackage[hypcap=true]{caption}
\title{\papertitle}

\affiliation{
\paperauthorA, \paperauthorB, \paperauthorC, and \paperauthorD \thanks{\vspace{-3mm}}}
{Department of Information Technology and Electrical Engineering \\ ETH Zurich, Switzerland\\
{\tt \href{mailto:joroth@ethz.ch}{joroth@ethz.ch} | \href{mailto:domkeller@ethz.ch}{domkeller@ethz.ch} | 
\href{mailto:caoscar@ethz.ch}{caoscar@ethz.ch} | 
\href{mailto:studer@ethz.ch}{studer@ethz.ch}}
}

\begin{document}
\ifpdf 
  \DeclareGraphicsExtensions{.png,.jpg,.pdf}
\else  
  \DeclareGraphicsExtensions{.eps}
\fi


\maketitle

\begin{abstract}
Analog subtractive synthesizers are generally considered to provide superior sound quality compared to digital emulations. However, analog circuitry requires calibration and suffers from aging, temperature instability, and limited flexibility in generating a wide variety of waveforms. Digital synthesis can mitigate many of these drawbacks, but generating arbitrary aliasing-free waveforms remains challenging. 
In this paper, we present the \pmsynth, a hybrid digital-analog eight-voice polyphonic synthesizer prototype that combines the best of both worlds. 
At the heart of the synthesizer is the big Fourier oscillator (BFO), a novel digital very-large scale integration (VLSI) design that utilizes additive synthesis to generate a wide variety of aliasing-free waveforms.
Each BFO produces two voices, using four oscillators per voice. A single oscillator can generate up to $1024$ freely configurable partials (harmonic or inharmonic), which are calculated using coordinate rotation digital computers (CORDICs).
The BFOs were fabricated as $65$\,nm CMOS custom application-specific integrated circuits (ASICs), which are integrated in the \pmsynth to simultaneously generate up to $32\,768$ partials.
Four $24$-bit $96$\,kHz stereo DACs then convert the eight voices into the analog domain, followed by digitally controlled analog low-pass filtering and amplification.
Measurement results of the \pmsynth prototype demonstrate high fidelity and low latency. 
\end{abstract}

\section{Introduction}
\label{sec:intro}

Digital sound synthesis has many advantages over implementations with analog circuitry. 
Most notably, digital implementations are able to produce a wide variety of waveforms and do not suffer from temperature instabilities, aging, and component variations. 
However, aliasing is a ubiquitous nuisance in digital sound synthesis and specialized signal processing techniques are often necessary to combat such artifacts. 
For example, the work in~\cite{stilson_1996} proposes low-complexity methods to generate aliasing-free waveforms of classical analog synthesizers (e.g., rectangle, sawtooth, and triangle). Nonetheless, this method is unable to generate more complex waveforms.
In stark contrast, direct digital synthesis (DDS)~\cite{reinhardt_direct_1985} enables the generation of nearly arbitrary waveforms at very low complexity. Unfortunately, na\"ive DDS implementations generally suffer from aliasing. While aliasing can be reduced to a certain extent with oversampling followed by low-pass filtering, such an approach diminishes the complexity advantages of DDS. 
The work in~\cite{schanze1995} presents a method that is able to generate arbitrary alias-free single-period wavetable waveforms. This method, however, requires intricate trigonometric functions that must be calculated at high precision and is, thus, not well-suited for efficient software and hardware implementations.

An alternative sound-synthesis approach that eliminates aliasing altogether is to use additive synthesis~\cite{moorer1977}, without ever generating partials that exceed half the sampling rate.
Many software synthesizers support additive synthesis and benefit from the flexibility and user-interface capabilities that software brings.
However, relying on general-purpose processors, such implementations have to balance signal quality and computational complexity, which limits the amount of partials that can be generated and affects their purity.\footnote{The maximum number of partials varies significantly across different plugins and depends on several parameters, such as the number of voices, signal quality, and others. Alchemy from Apple's Logic Pro X, for example, supports up to $600$ partials~\cite{alchemy_additive}, while others support less than a dozen.}
To overcome the limitations of additive synthesis software implementations, the work in~\cite{debernardinis99} proposes a specialized hardware design, which is able to generate a large number of partials (up to $1200$) with a single application-specific integrated circuit (ASIC).  
Such an implementation would enable the generation of a wide variety of complex and aliasing-free waveforms, but, to the best of our knowledge, no working system was demonstrated. 

In recent years, a number of commercially available hybrid digital-analog instruments emerged, which can generate a broad range of high-quality waveforms. 
Specific instances are the Arturia Freak Series~\cite{arturiafreak}, Sequential Prophet~X~\cite{prophetx}, Udo Super~6~\cite{udosuper6}, and Waldorf Quantum~\cite{waldorfquantum}. Unfortunately, only very little  is known about the inner workings of these instruments and, thus, it remains largely a mystery how the waveforms are synthesized.

\subsection{Contributions}

We present the \pmsynth, an eight-voice polyphonic hybrid digital-analog synthesizer prototype that combines digital oscillators with analog filtering and amplification. 
Each voice consists of four digital oscillators, each able to generate a wide range of aliasing-free waveforms using an additive synthesis approach\footnote{Thus the name \emph{\pmsynth}, where $\boldsymbol{+}$ represents the additive synthesis approach and $\boldsymbol{-}$ the subtractive architecture of the instrument.} with up to $1024$ partials (harmonic, inharmonic, or subharmonic); in total, the instrument can generate $32\,768$ partials simultaneously. 
The oscillators are implemented using a custom ASIC, the \emph{big Fourier oscillator (BFO)}, which generates two voices. 
To arrive at high hardware efficiency of the BFO ASIC implementation, we present a range of algorithm-level optimizations and a very-large scale integration (VLSI) architecture that utilizes CORDICs (short for coordinate rotation digital computers) to generate   partials of high purity.
We show how the BFO ASICs are integrated into the \pmsynth hardware prototype, including the analog section that incorporates voltage controlled filters (VCFs) and voltage controlled amplifiers (VCAs) based on commercial integrated circuits (ICs).
Finally, we present implementation results of the BFO ASIC and measurement results at various output stages in the instrument. 
A photo of the \pmsynth hardware prototype, including the external MIDI keyboard and controller, is shown in \fref{fig:synthsetup}.

\begin{figure}[tp]
\centerline{\includegraphics[width=0.9\columnwidth]{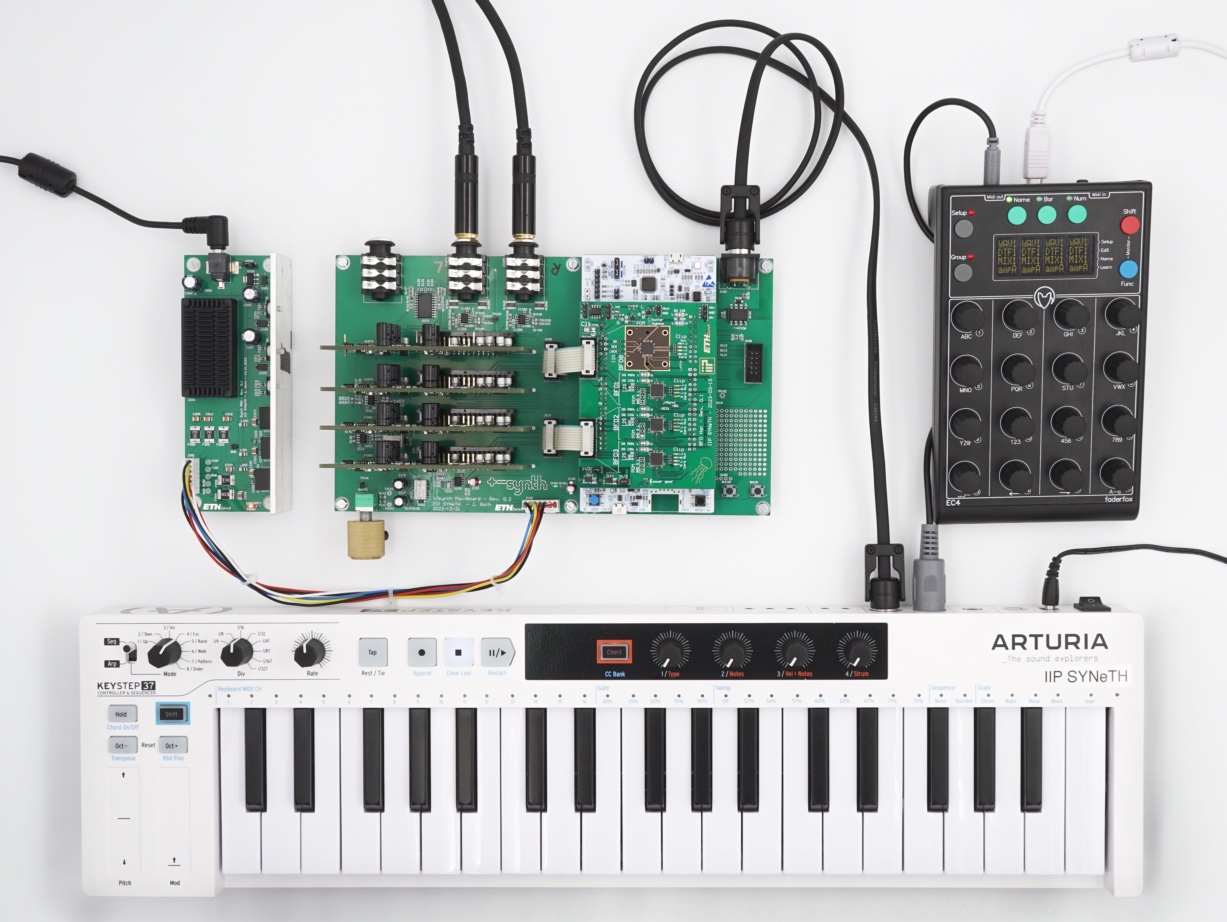}}
\caption{\label{fig:synthsetup}{\it Overview of the synthesizer setup. Top left: custom power supply; top middle: \pmsynth hardware prototype; top right: MIDI controller; bottom: MIDI keyboard.}}
\end{figure}

\section{Synthesizer Architecture}

Figure~\ref{fig:synth_standalone} shows a photo of the \pmsynth hardware prototype consisting of a main board (with audio, MIDI, and power connectors), an STM32 Nucleo development board (in white), with a custom printed circuit board (PCB) on top hosting four BFO ASICs, and four analog voice PCBs plugged into the main board (top right). The key components of the hardware design are discussed next.

\subsection{System Overview}

A system architecture overview of the  \pmsynth is given in \fref{fig:synth_overview}. 
The instrument relies on a hybrid digital-analog version of subtractive synthesis with digital oscillators, analog filters, and analog amplifiers. 
The \pmsynth is digitally controlled by a microcontroller unit (\mcu), which receives user inputs from an external MIDI keyboard and a universal MIDI controller with rotary encoders. The \mcu generates all of the control signals for the digital oscillators as well as for the analog filters and amplifiers.
The instrument is able to generate eight independent voices, where each BFO ASIC implements two voices and each voice consists of four digital oscillators. 
The digital voices are converted into the analog domain using four stereo DACs and each voice is separately processed by a VCF and a VCA; an \emph{analog voice PCB} houses these analog components required for two voices. The control voltages (CVs) are generated using CV-DACs, which are digitally controlled by the \mcu. 
The analog voices are summed to create the instrument's output signal, which can be either mono (eight voices) or stereo (four voices per channel); this is controlled by relays. The \pmsynth offers stereo line-level and headphone outputs.
A custom power supply derives all of the required voltages for both the analog and digital domains from an off-the-shelf $24$\,V DC power supply.

\begin{figure}[t]
\centering
\includegraphics[width=0.9\columnwidth]{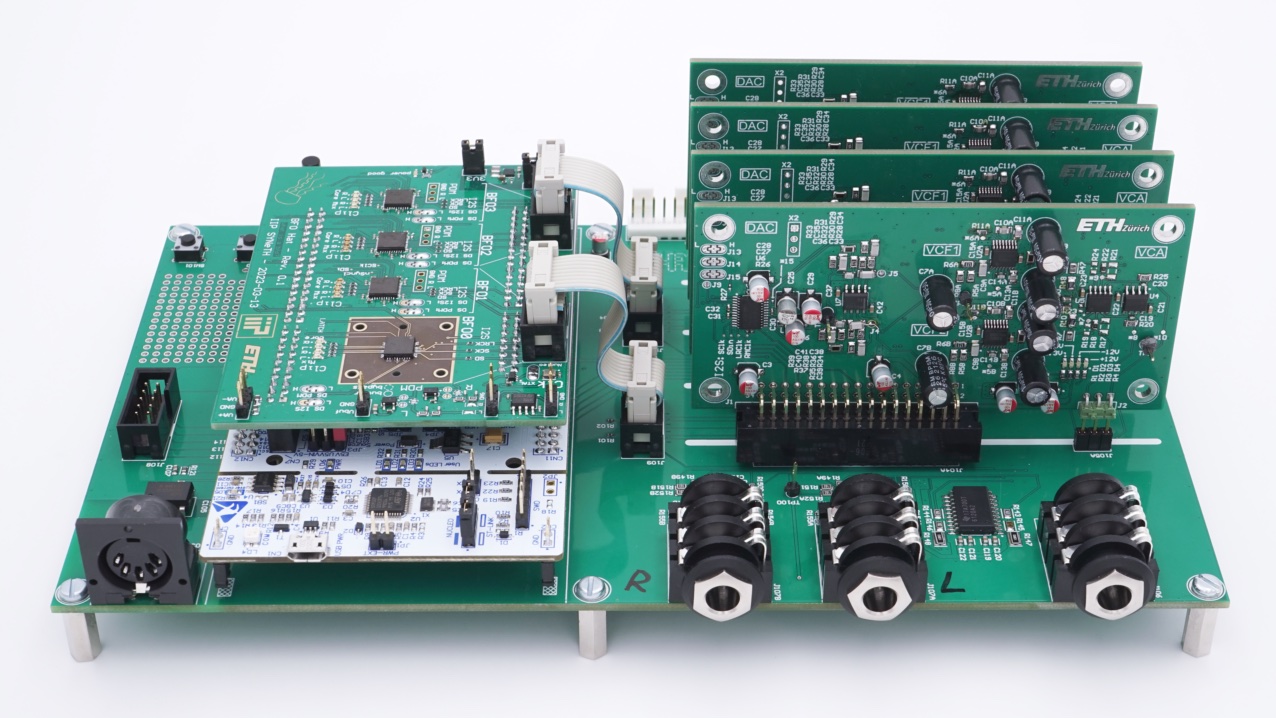}
\caption{\it Photo of the \pmsynth hardware prototype.}
\label{fig:synth_standalone}
\end{figure}

\begin{figure}[t]
\centering
\includegraphics[width=0.9\columnwidth]{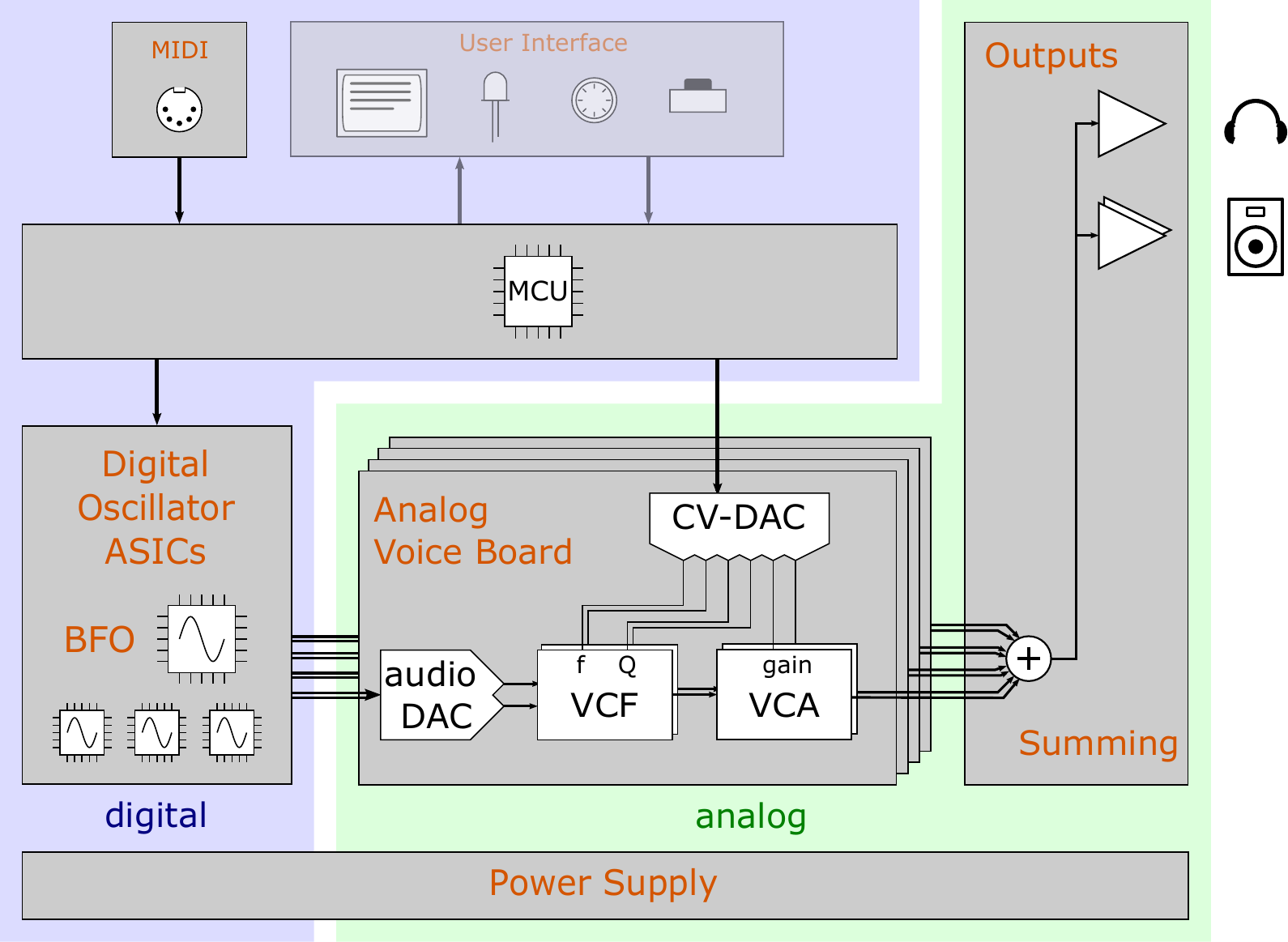}
\caption{\it Architecture overview of the \pmsynth.}
\label{fig:synth_overview}
\end{figure}

\subsection{Digital Control and Interfaces}

Digital control of the instrument is carried out on an STM23 Nucleo development board (STM32F722ZE), which was chosen to bypass the recent chip shortage. 
Received MIDI messages are used to compute the control signals for the audio path (e.g., oscillator pitch and volume, filter cutoff-frequency and resonance, and amplifier gain). The \mcu also implements the attack, decay, sustain, and release (ADSR) envelope generators as well as low-frequency oscillators (LFOs).
The \mcu generates and transmits the control signals to the digital oscillators using a serial parallel interface (SPI) bus. The oscillators then generate eight voices that are converted to the analog domain via inter-IC sound (I2S) interfaces and four $24$-bit stereo DACs. 
A second SPI bus on the \mcu interfaces with the CV-DACs that control the analog voice processing paths.
The details of the digital oscillators and analog voice PCBs are discussed next. 

\subsection{Aliasing-Free Digital Oscillator}
\label{sec:oscillator}
At the heart of the \pmsynth are digital aliasing-free oscillators, which utilize additive synthesis, each generating up to $K=1024$ partials.\footnote{With up to $K=1024$ partials per oscillator, we can generate, for example,  a sawtooth wave with a base frequency $f=20$\,Hz and the highest harmonic at $20\,480$\,Hz, which is at the edge of the audible spectrum.} Each voice is composed of four digital oscillators, which can be mixed together arbitrarily with configurable gains. Our custom ASIC, the BFO, implements two voices, and the  \pmsynth consists of four BFO ASICs to provide eight-voice polyphony. 
Each oscillator calculates its samples $x[\ell]$ using the following Fourier series: 
\begin{align} \label{eq:fourierseries}
x[\ell] = \sum_{k=1}^{K} a_{k}\cos\!\left(2\pi \frac{f n_{k}}{\fs} \ell\right)+b_{k}\sin\!\left(2\pi \frac{f n_{k}}{\fs} \ell\right)\!.
\end{align}
The oscillator parameters $\{a_k,b_k,n_k\}_{k=1}^K$, together with the oscillator's base frequency $f$ and the system's sampling rate $\fs$, fully determine the waveform to be generated. Note that each oscillator has its own set of parameters. 
In order to avoid aliasing altogether, we only sum the terms in~\fref{eq:fourierseries} indexed by the set\footnote{We note that aliasing can still occur if one reconfigures the oscillator parameters $\{a_k,b_k,n_k\}_{k=1}^K$ at too fast rates.} 
\begin{align} \label{eq:nyquistset}
\setK(f) = \{ k=1,\ldots,K : f n_k < \fs/2\},
\end{align} 
i.e., we replace $k=1,\ldots,K$ with $k\in\setK(f)$ in \fref{eq:fourierseries}.
\fref{sec:BFOVLSI} details how this additive synthesis approach is implemented in our custom BFO ASICs.

\begin{remark}
We emphasize that \fref{eq:fourierseries} is, strictly speaking, not a Fourier series, as we also allow the multipliers $n_k\in\mathbb{Q}$, $k=1,\ldots,K$, to be nonnegative rational numbers represented in the chosen fixed-point format (see \fref{sec:fixpoint} for the details). This flexibility enables us to generate waveforms with harmonics, inharmonics, and subharmonics, which implies that a single oscillator cannot only generate standard analog synthesizer waveforms, but also arbitrary wavetable sounds or bell-like timbres. 
 \end{remark}

Each of the four oscillators of a voice generates samples according to~\fref{eq:fourierseries}, which are weighted and summed. Then, the two BFO voices can be further mixed before being passed to the I2S output. Details on the mixing stage are provided in \fref{sec:bellsandwhistles}. 

\subsection{Analog Voice Boards}

\begin{figure}[t]
\includegraphics[width=0.9\columnwidth]{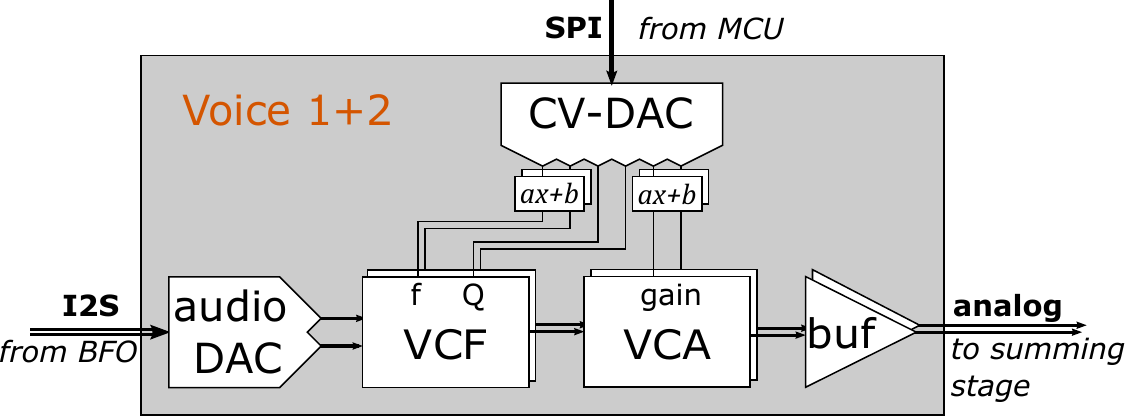}
\caption{{\it Block diagram of an analog voice PCB.}}
\label{fig:asv_diagram}
\end{figure}

The \pmsynth features four analog voice PCBs, which can be seen at the top-right of \fref{fig:synth_standalone}. Each of these PCBs carries out analog signal processing for two voices; a block diagram is depicted in \fref{fig:asv_diagram}.
The audio DAC receives the samples from the two voices generated by a BFO ASIC at a sampling rate of $96$\,kHz. 
We used a {CS4350} $24$-bit stereo DAC from Cirrus Logic~\cite{DatShCS4350},  which contains an integrated phase-locked loop (PLL) that derives its master clock from the I2S clock. This eliminates the need to route a  separate $25$\,MHz clock signal to each analog voice PCB.

The VCFs are implemented using the {SSI2144} IC from Sound Semiconductor~\cite{DatShSSI2144}, which implements a fourth-order low-pass ladder filter. The cutoff frequency and resonance (Q-factor) are set by CVs.
The VCAs are implemented using the  {SSI2162} IC from Sound Semiconductor~\cite{DatShSSI2162}, which is used to apply the envelope determined by another CV. Finally, the analog voice signal is buffered before leaving the PCB to the final summing stage on the main board that produces headphone and line outputs. 
Note that some CVs are low-pass filtered, scaled, and offset in order to avoid crosstalk to the audio path and to match the required voltage range; this is implemented using operational amplifiers.

\section{VLSI Design of the Big Fourier Oscillator}
\label{sec:BFOVLSI}

In order to develop an efficient VLSI design that is able to  implement multiple aliasing-free oscillators, each with a large number of high-quality partials, we leverage a range of algorithm- and hardware-level tricks which are discussed next.  

\subsection{Fixed-Point Arithmetic}
\label{sec:fixpoint}
Our VLSI design exclusively uses fixed-point arithmetic, mostly with $32$-bit fixed-point precision, which enables the generation of extremely pure partials with low harmonic distortion at high hardware efficiency. 
The fixed-point number format is designated by the notation $\{\square,\qi,\qf\}$, where $\qi$ is the number of integer bits, $\qf$ is the number of fractional bits, and $\square$ is either~$s$ or~$u$ if the fixed-point number is signed or unsigned, respectively.
The coefficients~$a_k$ and~$b_k$  in~\fref{eq:fourierseries} use the format $\{s,0,31\}$, which gives precise control over the partials' amplitudes and phases. The multipliers $n_k$ use the format $\{u,16,16\}$, which enables harmonic, inharmonic, and subharmonic partials with fine frequency resolution. 
The base frequency is normalized as $f/\fs$ and uses the format $\{u,0,32\}$, which results in a frequency resolution of $22.4$\,$\upmu$Hz at   $\fs=96$\,kHz.
While the generated samples have an internal resolution of $32$\,bits, the I2S output is reduced to the DACs' resolution of $24$\,bits. 

\subsection{Algorithm-Level Optimizations}
\label{sec:algo_opt}

To improve the hardware efficiency of our VLSI design, we use the following algorithm-level optimizations that reduce the complexity of calculating the samples as in \fref{eq:fourierseries}. 

\paragraph{Reparametrization from Radians to Turns}
Instead of directly computing the arguments $\phi_k[\ell] \define 2\pi \frac{f n_{k}}{\fs} \ell$ of the cosine and sine functions in \fref{eq:fourierseries}, which are in the unit of radians, our VLSI design calculates the two functions $\mathrm{co}(\theta) \define \cos(2\pi \theta)$ and $\mathrm{si}(\theta) \define \sin(2\pi \theta)$ instead. Here, the arguments $\theta_k[\ell] \define \frac{f n_{k}}{\fs} \ell$ are represented in what is known as \emph{turns}, which has several advantages.
First, the arguments $\theta_k[\ell]$ are in the range~$[0,1)$, which requires one to pass arguments to the $\mathrm{co}(\cdot)$ and $\mathrm{si}(\cdot)$ functions of the form
\begin{align} \label{eq:thetak}
\theta_k[\ell]=\frac{f n_{k}}{\fs} \ell \bmod{1}.
\end{align}
The modulo-$1$ operation can be obtained in hardware for free by simply discarding the integer part of~$\theta_k[\ell]$ when represented by unsigned fixed-point numbers. 
Second, one can directly calculate the functions $\mathrm{co}(\cdot)$ and $\mathrm{si}(\cdot)$ in a hardware-friendly manner using CORDICs; see \fref{sec:CORDIC} for the details.

\paragraph{Sequential Calculation of Arguments}
For each sample $\ell=0,1,\ldots$ and partial $k=1,\ldots,K$, two multiplications are required to calculate the argument $\theta_k[\ell]$ in \fref{eq:thetak}. First, the argument increment $\delta_k \define {f n_{k}}/{\fs}$ is computed by multiplying the normalized base frequency $f/\fs$ by the frequency multiplier $n_k$.
Second, the argument increment $\delta_k$ is multiplied by the sample index $\ell$.
While this last multiplication occurs in \mbox{modulo-$1$} arithmetic, it still requires a high dynamic range, as the sample index $\ell$ is represented with a large number of bits to avoid unwanted resetting in the oscillators.
Hence, every sample $\ell$ requires $K=1024$ of these high-resolution $\theta_k[\ell]=\delta_k \ell \bmod{1}$ products. 
This complexity could easily be reduced by tracking the argument~$\theta_k$ for each partial $k=1,\ldots,K$ with an addition instead of a multiplication. In specific, one can update the argument~$\theta_k$ for every new sample $\ell$ as  
\begin{align} \label{eq:iterativeargumentupdate}
\theta_k \gets (\theta_k + \delta_k) \bmod{1}.
\end{align}
The arguments are initialized as $\theta_k=0$ at sample index $\ell=0$. 
We reiterate that the modulo-$1$ operation is free in hardware by discarding the integer part of the arguments $\theta_k$, $k=1,\ldots,K$.  

\paragraph{Single Argument for All Partials}
The remaining disadvantage of the above method is that one must keep track of the arguments $\theta_k$ for every partial $k=1,\ldots,K$. This requires additional storage for all arguments in a two-port memory that supports one simultaneous read and write per update of \fref{eq:iterativeargumentupdate}. 
To avoid an additional memory, we keep track of a \emph{single} base-frequency argument $\theta\define\frac{f}{\fs} \ell$  from which all arguments~$\theta_k$ can be derived as follows:
\begin{align} \label{eq:argumentutok}
\theta_k = \theta n_k \bmod{1}, \quad k=1,\ldots,K.
\end{align}
Unfortunately, unlike the arguments $\theta_k$ in \fref{eq:iterativeargumentupdate}, the base-frequency argument $\theta$ cannot be accumulated modulo-$1$ since the property 
\begin{align}\label{eq:keyproperty}
\theta n_k \bmod{1}=(\theta\bmod{m}) n_k \bmod{1} 
\end{align}
with $m=1$ does \emph{not} hold for every $n_k$.
For example, for $\theta=1$ and $n_k=1.5$, $(\theta\bmod{1})n_k\bmod{1}=0$ is different from the desired $\theta n_k\bmod{1}=0.5$.
Indeed, to calculate $\theta_k$ for an arbitrary frequency multiplier $n_k$ from the base-frequency argument~$\theta$, the quantity $\theta$~needs to be represented with an infinite number of integer bits (i.e., without using a modulo operation).
Nevertheless, provided that the multipliers~$n_k$ are represented with a finite number of fractional bits $\qf$, it is sufficient to represent $\theta$ using a finite number of integer bits. This key insight is made rigorous by \fref{lem:modchange}.

\begin{lem}
\label{lem:modchange}
 $\theta n \bmod{l}=(\theta\bmod{m}) n \bmod{l}$ if $m n \bmod l = 0 $.
\end{lem}
\begin{proof}
Let $\Theta=\theta \bmod m$, so that we want to show $\theta n \bmod{l}=\Theta n \bmod{l}$. We rewrite $\Theta=\theta \bmod m$ and $mn \bmod l = 0$ as
\begin{align}
\theta & = a m + \Theta,~\text{and} \label{eq:modtheta} \\
m n & = bl, \label{eq:modell}
\end{align}
where $a,b\in\mathbb{Z}$. Multiplying both sides of \fref{eq:modtheta} by $n$, we get
\begin{align}
\theta n & = a m n + \Theta n = a b l + \Theta n, \label{eq:modellin}
\end{align}
where \fref{eq:modellin} follows from \fref{eq:modell}.
Since $a b\in\mathbb{Z}$, taking the modulo-$l$ of both sides of \fref{eq:modellin} results in $\theta n\bmod{l} = \Theta n \bmod l$.
\end{proof}

By applying \fref{lem:modchange} with $n=n_k$ and $l=1$, we can determine the value of $m$ so that \fref{eq:keyproperty} is satisfied for the multipliers $n_k$ used in the BFO. In words, \fref{lem:modchange} is telling us that, instead of keeping track of $\theta$ with infinite precision, we can just keep track of its modulo-$m$ equivalent and any multiplication by the frequency multipliers~$n_k$ will be correct in modulo-$1$ arithmetic as long as $m n_k\in\mathbb{Z}$. Given that $n_k$ has $\qf=16$ fractional bits, we can satisfy this last requirement by setting $m=2^{\qf}$. Therefore, we keep track of the base-frequency argument~$\theta$ with the recursion
\begin{align} \label{eq:argumentu}
\theta\gets(\theta+\delta) \bmod{2^{\qf}},
\end{align}
where $\delta=f/\fs$. 
The modulo-$2^{\qf}$ operation is easily implemented in hardware by using only $\qf$ bits to represent the integer part of $\theta$ and letting the result wrap-around once the maximum representable number is reached.
We note that, if all $n_k\in\mathbb{Z}$, then $\theta$ could be tracked in modulo-$1$ arithmetic---thus, the recursion in \fref{eq:argumentu} is required as we also support inharmonic and subharmonic partials.

We observe that, while this approach requires $K=1024$ multiplications per sample~$\ell$ as in \fref{eq:argumentutok}, it avoids (i) storing arguments per partial and (ii) additional $K=1024$ multiplications per sample that would be required to compute the frequency increments $\delta_k=\delta n_k$. Thus, our approach leverages a trade-off between the high complexity of na\"ively computing $\theta_k[\ell]$ as in \fref{eq:thetak} and the large memory overhead of a per-$\theta_k$-argument accumulation as in~\fref{eq:iterativeargumentupdate}.

\subsection{Computing Cosines and Sines with CORDICs}
\label{sec:CORDIC}
To calculate cosine and sine functions at high precision and in a hardware-friendly way, we utilize CORDICs~\cite{volder}, which essentially calculate two-dimensional Givens rotations of the following form:
\begin{align} \label{eq:givens}
\underbrace{
\left[\begin{array}{c}
p_1' \\
p_2'
\end{array}\right]}_{\bmp'} = 
\underbrace{
\left[\begin{array}{cc}
\cos(\phi) & -\sin(\phi)\\
\sin(\phi) & \cos(\phi)
\end{array}\right]}_{=\bG(\phi)}
\underbrace{
\left[\begin{array}{c}
p_1 \\
p_2
\end{array}\right]}_{=\bmp}\!.
\end{align}
Evidently, by setting $\phi = 2\pi \frac{f n_{k}}{\fs} \ell$, $p_1=a_k$, and $p_2=-b_k$, the output~$p_1'$ is exactly one term of the Fourier series in \fref{eq:fourierseries}. 

We now outline the idea behind CORDICs---the interested reader is referred to~\cite{parhami} for more details. 
First, one approximates the desired rotation angle $\phi\approx\sum_{m=0}^{M-1} \phi_m$ by summing $M$ predefined micro-rotation angles $\phi_m$, $m=0,\ldots,\mbox{$M-1$}$; see below for a concrete choice of these angles.
Second, the Givens rotation in~\fref{eq:givens} is approximated by $\bG(\phi)\approx\prod_{m=0}^{M-1}\bG(\phi_m)$ with~$M$ so-called micro-rotations $\bG(\phi_m)$, which are simply Givens rotations by the angles~$\phi_m$. 
Third, one rewrites each micro-rotation as 
\begin{align} 
\bG(\phi_m) = \kappa_m
\left[\begin{array}{cc}
1 & -\tan(\phi_m)\\
\tan(\phi_m) & 1
\end{array}\right]\!,
\end{align}
where $\kappa_m=(1+\tan^2(\phi_m))^{-\frac{1}{2}}$. 
Fourth, one restricts the micro-rotation angles~$\phi_m$ to $\tan(\phi_m)=d_m2^{-m}$ with $d_m\in\{-1,+1\}$.
With this, the Givens rotation in \fref{eq:givens} is approximated as
\begin{align} \label{eq:microrotations}
\bmp' \approx \kappa \prod_{m=0}^{M-1} 
\left[\begin{array}{cc}
1 & -d_m 2^{-m} \\
d_m 2^{-m} & 1
\end{array}\right] \bmp.
\end{align}
Here, the scaling factor $\kappa=\prod_{m=0}^{M-1}\kappa_m$  depends only on the number of micro-rotations $M$ and not on the choices of $d_m$.  
Fifth, one needs to determine the micro-rotation angles $\phi_m$ that well-approximate the target angle $\phi$.
The standard procedure iteratively determines $\phi_m$  from the target angle $\phi$, i.e., by first taking the angle $\phi_0=d_0\mathrm{atan}(2^{-0})$ that brings $\phi$ closer to zero. One then updates the target angle as $\phi\gets\phi-d_0\mathrm{atan}(2^{-0})$ and  takes a new angle $\phi_1=d_1\mathrm{atan}(2^{-1})$ that brings the updated $\phi$ closer to zero.
This procedure is repeated for the remaining $M-2$ micro-rotations. 

We note that every additional micro-rotation provides roughly one additional bit of precision~\cite{parhami}; this implies that the approximation error in \fref{eq:microrotations} can be made arbitrarily small. In our application, we use a quadrant correction followed by $M=26$ micro-rotations, which leads to a precision of $\qf=24$ fraction bits (see \fref{sec:measurements} for measurements).
Furthermore, instead of representing the micro-rotation angles $\phi_m$ in radians, we represent them in turns; this enables us to directly calculate the functions $\mathrm{co}(\cdot)$ and $\mathrm{si}(\cdot)$ with a CORDIC. 
Finally, it is crucial to realize that each micro-rotation in  \fref{eq:microrotations}  only involves shifts, additions, subtractions, and a multiplication by the constant $\kappa$; this implies that the Givens rotation in \fref{eq:givens} can be approximated in a hardware-friendly manner, generating samples of cosine and sine functions with extremely high purity.

\begin{figure}[tp]
\centering
\begin{minipage}[b]{0.4425\columnwidth}
\vspace{0pt}
\includegraphics[width=1\columnwidth]{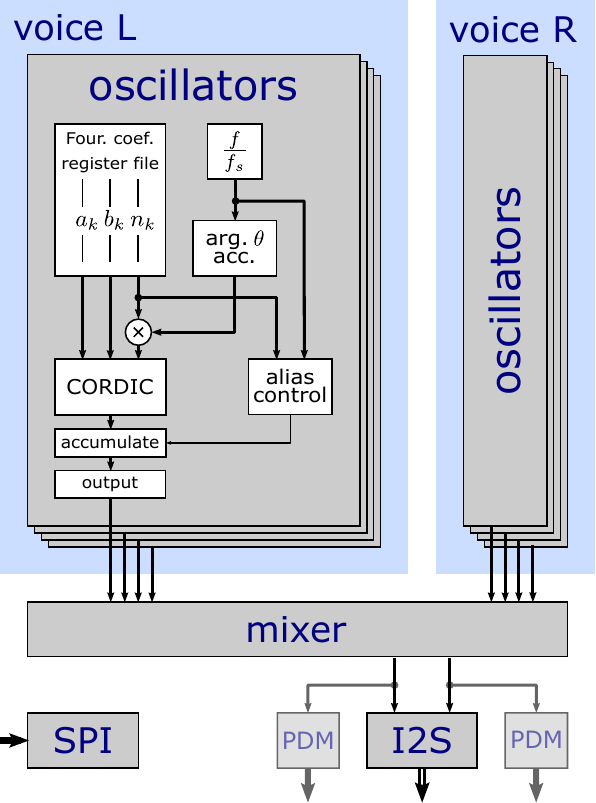}
\caption[caption]{\it BFO architecture.}
\label{fig:bfo_overview}
\vspace{0pt}
\end{minipage}
\hfill
\begin{minipage}[b]{0.45\columnwidth}
\includegraphics[width=1\columnwidth]{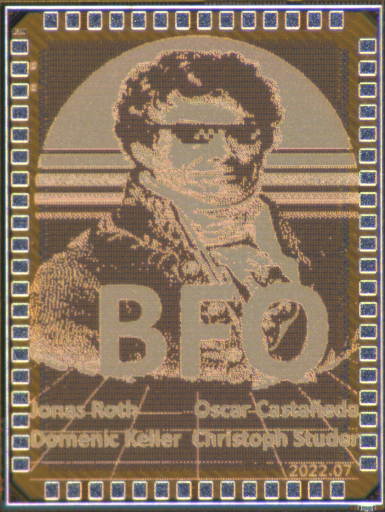}
\caption[caption]{\it ASIC micrograph.}
\label{fig:bfo_die}
\vspace{0pt}
\end{minipage}
\end{figure}

\subsection{VLSI Architecture}
\label{sec:VLSIarchitecture}

\fref{fig:bfo_overview} depicts the VLSI architecture of our BFO ASIC, which consists of two voices, referred to as L and R, each one with four oscillators.
Each oscillator features a $1024\times 96$-bit register file to store the oscillator coefficients $\{a_k,b_k,n_k\}$, $k=1,
\ldots,1024$, an argument accumulator, a fully unrolled CORDIC module that calculates one pair of sine and cosine values per clock cycle, and a sample accumulator.
The configuration registers and register files of the BFO are addressed using a memory map and are configured via SPI. Each SPI command uses $48$\,bits: $16$\,bits for the address and $32$\,bits for the data. The SPI interface runs at a baud rate of $13$\,Mb/s, with which a new waveform with $K=1024$ partials could be reprogrammed in less than $12$\,ms---nevertheless, this task is currently completed in $315$\,ms as our  firmware does not yet fully exploit the \mcu's capabilities.

To generate samples $x[\ell]$ as in~\fref{eq:fourierseries} with $K=1024$ partials  at a sampling rate $\fs=96$\,kHz, each oscillator operates at a clock frequency of $K\fs=98.304$\,MHz and calculates one sample of one partial every clock cycle.
The argument accumulator uses the base frequency $f/\fs$, stored in a configuration register, to track the base-frequency argument $\theta$, which is updated every $K=1024$ clock cycles as in \fref{eq:argumentu}.
In each of the $K$ clock cycles that the base-frequency argument $\theta$ remains fixed, one word of the register file is read to obtain $\{a_k,b_k,n_k\}$.
Then, the argument $\theta_k$ is computed by multiplying $\theta$ and $n_k$ as in \fref{eq:argumentutok}. 
With $\{a_k,b_k,\theta_k\}$ available, the CORDIC computes one partial in \fref{eq:fourierseries}, which is then accumulated to the current sample $x[\ell]$ if aliasing will not occur, i.e., if $fn_k<\fs/2$ is met. 
After $K=1024$ clock cycles, the four samples generated by the four oscillators are added to create one voice sample, and the L and R voice samples are then mixed using a programmable $2\times 2$ matrix to support, e.g., stereo and mono processing.
The two mixed samples are then streamed to the DACs via I2S.

\fref{fig:bfo_die} shows a micrograph of the $3$\,mm$^2$ BFO ASIC, which was fabricated in TSMC $65$\,nm LP CMOS technology. At the nominal $1.2$\,V core supply and room temperature, the ASIC achieves a maximum measured clock frequency of $154$\,MHz, exceeding the required $98.304$\,MHz, while consuming only $178$\,mW.

\subsection{Bells and Whistles}
\label{sec:bellsandwhistles}

The BFO includes a number of additional features, which further improve its versatility and flexibility. These features are as follows. 

\paragraph{Subwave Mixing} 
Per default, each oscillator accumulates $1024$ partials together to compute one sample $x[\ell]$ as in \fref{eq:fourierseries}.
We also support a \emph{subwave mixing} mode, in which the $1024$ partials are split into up to four disjoint groups (or \emph{subwaves}) that are accumulated independently.
By doing so, a single oscillator can blend various wavetable sounds with the same base frequency, each one with fewer partials; e.g., an oscillator can generate four subwaves of $256$ partials each instead of a single waveform with $1024$ partials.
Each subwave $x_i[\ell]$, $i=1,\ldots,4$, has a corresponding weight $v_i$, so that the output of one oscillator is $y[\ell]=\sum_{i=1}^{4} v_i x_i[\ell]$.
Thus, with the four oscillators per voice, subwave mixing can arbitrarily blend up to $16$ different wavetables in a single voice.

\paragraph{Aliasing Control} 
Per default, only partials for which ${f}n_k/\fs  < 0.5 $ holds are accumulated; see~\fref{eq:nyquistset}.
The BFO includes a mode that accumulates partials for which the following condition is met
\begin{align}
f_\text{HP} \leq {f}n_k/{\fs} < f_\text{LP},
\end{align}
which realizes an ideal band-pass filter with the lower and upper cutoff frequencies $f_\text{HP}$ and $f_\text{LP}$, respectively. These frequencies can be configured in the range $[0,1)$; the default values are $f_\text{HP} = 0$ and $f_\text{LP} = 0.5$ (no aliasing allowed).
Having these frequencies programmable can be used, e.g., to intentionally allow for aliasing (by setting $f_\text{LP} > 0.5$), or to apply low- and high-pass filtering. 

\paragraph{Bit- and Rate-Crusher}
Each oscillator output includes (optional) \emph{bit-crusher} and \emph{rate-crusher} distortion effects.
The bit-crusher forces certain bits of the output sample~$y[\ell]$ to zero; the zero bits are determined by a programmable bit-mask.
This feature can be used to reduce the bit-resolution by zeroing a certain number of least-significant bits; other, more complicated masking patterns are also possible. 
The rate-crusher is a simple sample-and-hold sub-sampling circuit that lowers the rate at which the output samples~$y[\ell]$ are updated.
The sub-sampling rate is set in the range $(0, \fs)$ and is stored in a configuration register.
This feature can be used to emulate lower sampling rates or, for example, cause intriguing aliasing artifacts that depend on the base frequency $f$. 
 
 \paragraph{PDM Output}
As an alternative to the I2S interface, each BFO ASIC also includes two pulse density modulation (PDM) outputs. 
The 1-bit PDM signal is generated from a digital first-order sigma-delta modulator running at an oversampling rate of $1024$, which corresponds to the BFO's clock frequency. 
This output could be useful for cost-sensitive applications (as no DACs are required), but at reduced sound quality; see \fref{sec:signalquality} for measurement results. 
Note that the PDM outputs are not used in the \pmsynth prototype.

\paragraph{Clipping Indicators}
At certain stages in the digital signal processing path, sample values must be clipped to a certain maximum range to reduce their word lengths. This occurs at the output of each oscillator and after the mixer stage. 
Thus, we included a number of clipping indicator outputs that are raised if clipping occurred; these are tied to LEDs on the \pmsynth prototype. 

\subsection{Comparison}

\begin{figure}[tp]
\centerline{\includegraphics[width=0.95\columnwidth]{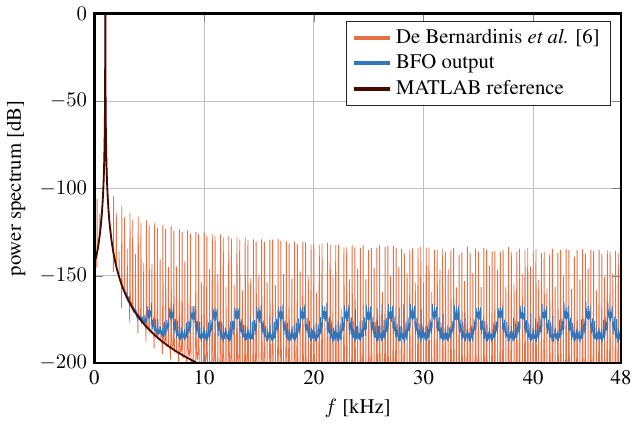}}
\caption{\label{fig:spectrum_osc}{\it Power spectrum comparison for a $1$\,kHz sine wave at a sampling rate of $\fs=96$\,kHz between a MATLAB floating-point reference, the BFO output, and the method described in~\cite{debernardinis99}.}}
\end{figure}

To the best of our knowledge, only the ASIC design reported in~\cite{debernardinis99} is comparable to our BFO implementation.
The design in~\cite{debernardinis99} utilizes a marginally stable infinite-impulse response (IIR) filter to recursively generate samples of a single sinusoid only with the help of a multiply-accumulate unit. 
While the algorithm itself is competitive to our approach in complexity per generated sample, the recursive calculations together with fixed-point  arithmetic suffer from error propagation, particularly at low frequencies (in the order of tens of Hz). To mitigate this issue, the recursion must be restarted periodically (the authors recommend restarting every $128$ samples), which requires one to retrieve two sinusoids from two consecutive sample instants from a memory or a CORDIC. The design in \cite{debernardinis99} assumes that these initial sine values are generated externally. Moreover, their circuit assumes that the magnitudes and phases are stored externally and streamed into the ASIC. Thus, their design would require additional (external) logic and memory, whereas our BFO ASIC is fully self-contained.

A direct comparison between the hardware implementation characteristics of our BFO ASIC and the design in \cite{debernardinis99} is challenging due to missing details. 
Nonetheless, we re-implemented their method in MATLAB using the fixed-point parameters of~\cite{debernardinis99} and compared it to a double-precision floating-point MATLAB reference and the BFO output for a $1$\,kHz sine at $\fs=96$\,kHz. 
Our simulations reveal that the method in~\cite{debernardinis99} achieves a total harmonic distortion plus noise (THD+N) of $-94.27$\,dB, which is $42.7$\,dB worse than what is achieved by our BFO ASIC (see also~\fref{tbl:thdn}).
The corresponding power spectra\footnote{We analyze $960$\,k samples ($10$\,s) using Welch's method with a Hann window, a $2^{14}$-point FFT, $50$\% overlap, and a normalized peak value of $1$.} are shown in \fref{fig:spectrum_osc} and it is evident that our CORDIC-based approach generates sine waves with significantly higher purity than the ASIC design reported in~\cite{debernardinis99}.

\section{Measurement Results}
\label{sec:measurements}
In order to quantify the performance of the \pmsynth prototype, we now present a number of measurement results.

\subsection{Signal Quality}
\label{sec:signalquality}

We first assess the quality of a single partial at different stages of the instrument:
(i)~{the BFO output} (I2S output, digital domain),
(ii)~{the DAC output} (after reconstruction filter, analog domain), 
and (iii)~{the \pmsynth output} (line-level output, analog domain).
We generate a $1$\,kHz sine wave with  $-6$\,dBFS amplitude (with respect to the I2S interface) in the BFO at a sampling rate $\fs$ of~$96$\,kHz.
The digital BFO output is captured using a logic analyzer to extract raw I2S data; the analog signals are captured using a Focusrite Scarlett 18i20 (3rd gen.) audio interface~\cite{Scarlett18i20} with a sampling rate of $96$\,kHz. 
\fref{fig:spectrum_synth} shows the power spectrum of the three signals. 
The corresponding THD+N values are reported in \fref{tbl:thdn}. 
We see that the test signal (sinusoidal) at the BFO output has extremely high purity and the quality is essentially limited by the DAC. We can also see that analog processing through the VCF and VCA circuitry further reduces the THD+N, which is not unexpected. 

We also measured the filtered PDM output, which consists of a passive second-order low-pass filter with a $-3$\,dB frequency of $31$\,kHz. 
This output achieves a THD+N of $-75$\,dB, which is $13$\,dB higher than that of the DAC output (see~\fref{tbl:thdn}); the associated spectrum is shown in \fref{fig:dac_vs_pdm}. 
Clearly, the PDM output is inferior to the DAC output, but would enable the use of our BFO ASICs with less expensive external circuitry. 

\begin{figure}[t]
\centering
\subfigure[\it Power spectra at different stages in the \pmsynth.]{\includegraphics[width=0.95\columnwidth]{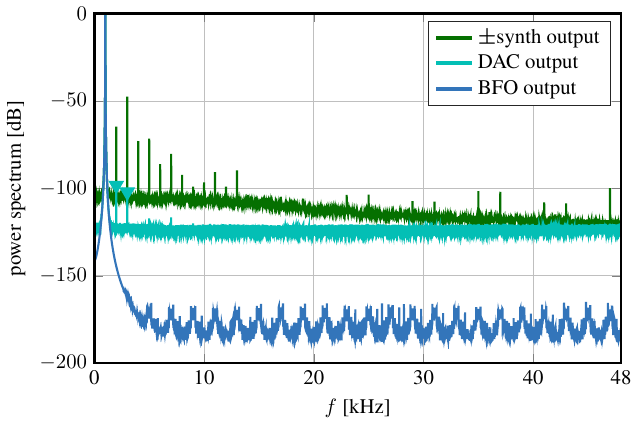}\label{fig:spectrum_synth}}
\subfigure[\it Power spectra of reference, DAC, and filtered PDM outputs.]{\includegraphics[width=0.95\columnwidth]{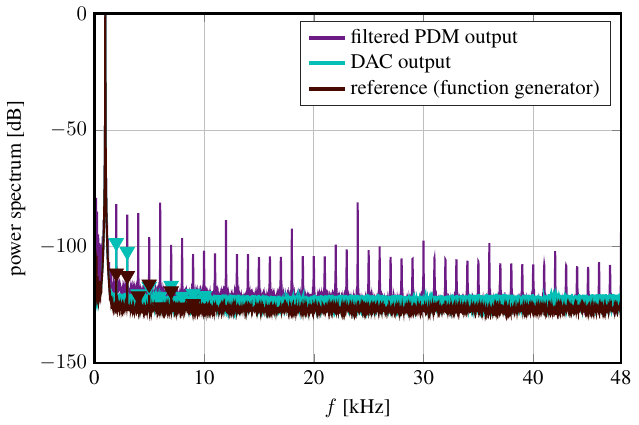}\label{fig:dac_vs_pdm}}
\caption{{\it Measured power spectra (in decibels) for a $1$\,kHz sine wave at a sampling rate of $\fs=96$\,kHz.}}
\end{figure}

\begin{remark}
The audio interface~\cite{Scarlett18i20} specifies a THD+N below $0.002\%$ ($\approx -94$\,dB) for the line inputs. 
As a reference, we measured a $1$\,kHz sine wave generated with an SRS DS360 ultra-low distortion function generator~\cite{SRSds360}, which resulted in a THD+N or $-89.5$\,dB (cf.~\fref{tbl:thdn}); the associated spectrum is shown in \fref{fig:dac_vs_pdm}. 
Since this THD+N result is close to that of the DAC output, our measurements are likely affected by the audio interface. 
\end{remark}

\begin{table}[tp]
\caption{\itshape THD+N measurements for a $f=1$\,kHz sine.}
\label{tbl:thdn}
\centering
\renewcommand{\arraystretch}{1.05}
\begin{tabular}{@{}lc@{}}
	\toprule
 	Measurement& THD+N [dB] \\
	\midrule
	\pmsynth output (analog) & $-51.0$ \\
	PDM output (analog) & $-75.0$ \\
	DAC output (analog) & $-88.2$ \\	
	BFO output (digital) & $-137.0$ \\
	\midrule
	reference: function generator (analog) & $-89.5$ \\
	\bottomrule
\end{tabular}
\end{table}

\fref{tbl:waves} shows the signal-to-noise and distortion ratio (SINAD) between a MATLAB floating-point model and the digital BFO output for different waveforms generated at a base frequency of \mbox{$f=20$\,Hz}.
Most waveforms achieve a SINAD exceeding $131$\,dB except for the pulse waveform, which yields $109.4$\,dB. The reason is that since the pulse waveform has the $a_k$-coefficients of all $1024$ partials set to the same value, that value has to be reduced substantially to avoid clipping. Thus, the signal power of this waveform is much lower compared to the other waveforms, which results in lower SINAD.

\begin{table}[tp]
\caption{\itshape SINAD for different waveforms at $f=20$\,Hz.}
\label{tbl:waves}
\centering
\renewcommand{\arraystretch}{1.05}
\begin{tabular}{@{}lc@{}}
	\toprule
 	Waveform& SINAD [dB] \\
	\midrule
	sine & 134.0 \\
	triangle & $133.3$ \\
	sawtooth & $135.3$ \\
	super-saw & $131.3$ \\
	rect-saw & $131.0$ \\
	pulse & $109.4$ \\
	\bottomrule
\end{tabular}
\end{table}

\subsection{Latency}
As any digitally controlled instrument, the \pmsynth exhibits non-negligible latency between a keystroke (or parameter change) and a change in the synthesizer's output.
To assess the prototype's latency, we measure the delay between the reception of a note-on MIDI message at the \pmsynth and the change in signal at the line-level output; this ignores external delays, e.g., caused by the MIDI keyboard.
The measurements are repeated $34$ times using an oscilloscope probing two signals: (i) the opto-coupler's output in the MIDI receiver circuit and (ii) the  hot signal of the line-output.
Our measurements show a mean latency of $2.09$\,ms (minimum $1.60$\,ms; maximum $2.76$\,ms; standard deviation $0.33$\,ms). 
The latency is mainly caused by the \mcu firmware, i.e., processing the UART data (MIDI receiver), computing new control signals, transmitting parameter data over SPI, etc.

\subsection{Power Consumption}
The \pmsynth's power consumption is measured for two cases: 
(i)~idle mode (default state after power-up) and (ii)~playback (all voices playing).
The supply current is measured at $499$\,mA during idle mode and up to $522$\,mA during playback. This corresponds to a power consumption of approximately $12$\,W to $13$\,W. 
Since there are many factors influencing the instrument's power (e.g., filter resonance, load at audio outputs, etc.), these measurements should be taken with a grain of salt.

\section{Limitations and Future Work}

The \pmsynth prototype, in its current form, has a number of limitations, which we now summarize.
First, the BFO ASICs are currently unable to perform oscillator synchronization, which is a direct consequence of the additive-synthesis approach discussed in~\fref{sec:oscillator}.
To mitigate this limitation, one could load in oscillator parameters $\{a_k,b_k,n_k\}_{k=1}^K$ of a synchronized oscillator, but this approach is limited by the rate at which all of the parameters can be rewritten (see \fref{sec:VLSIarchitecture}). Also, such a workaround might no longer be aliasing-free. Developing hardware-friendly solutions to implement true oscillator synchronization without aliasing, e.g., inspired by the works of \cite{brandt2001hard,la2022general}, is part of ongoing work.
Second, the BFO ASICs do not provide hardware support for true additive synthesis with separate envelopes per partial. Such functionality could readily be implemented in hardware, but comes at the cost of additional memory and logic to store and update the envelope parameters. A potential compromise would be to use the linear updates put forward in \cite{debernardinis99}. 
Third, the BFO does not provide a noise generator. We are planning to include such missing functionality in a future version. 
Fourth, the quality of the BFO is currently limited by the used $24$-bit DAC, and our audio interface may affect our measurements; in the future, we will use a better DAC to fully exploit the BFO's high purity and also use better measurement equipment.
Fifth, we are currently using an off-the-shelf external keyboard and controller---developing a dedicated user interface for the \pmsynth would be quite exciting.

\section{Conclusions}

We have shown the implementation details of the \pmsynth, an eight-voice hybrid digital-analog music synthesizer prototype that uses aliasing-free digital oscillators followed by analog filtering and amplification. 
By implementing the oscillators on custom ASICs, we are able to generate a wide variety of waveforms with up to $1024$ freely programmable partials per oscillator, which includes not only classical waveforms of analog synthesizers, but also wavetable sounds or bell-like timbres. 
The  \pmsynth is able to generate a total number of $32\,768$ partials at a sampling rate of $96$\,kHz, and measurement results have demonstrated high fidelity and low latency. 

While a range of commercial hybrid digital-analog synthesizers became available recently, virtually nothing is known about their inner workings.  
In contrast, every implementation detail of our hardware prototype is known and well-documented, and every aspect of the instrument can be modified easily.  
We therefore believe that the \pmsynth will be an excellent research platform for future real-time audio synthesis experiments.

\balance
\bibliographystyle{IEEEbib}

\end{document}